\documentclass[11pt]{article}

\usepackage[margin=1.25in]{geometry}
\usepackage[utf8]{inputenc} 
\usepackage[T1]{fontenc}    
\usepackage{hyperref}       
\usepackage{url}            
\usepackage{booktabs}       
\usepackage{amsfonts}       
\usepackage{nicefrac}       
\usepackage{microtype}      
\usepackage{amssymb,amsmath,amsthm,amscd,dsfont,mathrsfs,bbold,pifont,bm,mathtools}
\usepackage{blkarray}
\usepackage{graphicx,float,psfrag,epsfig,color}
\usepackage{comment}
\usepackage{subcaption}
\usepackage{algorithm}
\usepackage[noend]{algpseudocode}
\usepackage{enumitem}

\newtheorem{innercustomgeneric}{\customgenericname}
\providecommand{\customgenericname}{}
\newcommand{\newcustomtheorem}[2]{%
  \newenvironment{#1}[1]
  {%
   \renewcommand\customgenericname{#2}%
   \renewcommand\theinnercustomgeneric{##1}%
   \innercustomgeneric
  }
  {\endinnercustomgeneric}
}
\newcustomtheorem{customthm}{Theorem}
\newcustomtheorem{customlemma}{Lemma}
\newcustomtheorem{customdef}{Definition}
\theoremstyle{plain}
\newtheorem{definition}{Definition}
\newtheorem{thm}{Theorem}

\newtheorem{lemma}{Lemma}

\newcommand{\pluseq}{\mathrel{+}=}

\newcommand{\notni}{\not\mathrel{\text{\reflectbox{$\in$}}}}

\def\prob{{\mathbb P}}

\def\1{{\mathbf{1}}}

\def\ub{\textnormal{UB}}
\def\lb{\textnormal{LB}}
\def\nbub{\textnormal{NB-UB}}
\def\nblb{\textnormal{NB-LB}}
\def\inprob{\texttt{ProcessIncomingMsg}_\texttt{UB}}
\def\outprob{\texttt{GenerateOutgoingMsg}_\texttt{UB}}
\def\lowprob{\texttt{ProcessIncomingMsg}_\texttt{LB}}
\def\upinf{\sigma^+}
\def\lowinf{\sigma^-}
\def\to{\!\rightarrow\!}
\def\lowin{\textnormal{M}}
\def\incur{\textnormal{M}_{\textnormal{curr}}}
\def\innext{\textnormal{M}_{\textnormal{next}}}
\def\vcur{\textnormal{MSrc}}
\def\minsert{\textnormal{insert}}

\title{Nonbacktracking Bounds on the Influence in Independent Cascade Models}

\begin{document}

\author{
Emmanuel Abbe\thanks{Program in Applied and Computational Mathematics, and the Department of Electrical Engineering, Princeton University, Princeton,
USA, \texttt{eabbe@princeton.edu.}}
\and 
Sanjeev Kulkarni\thanks{The Department of Electrical Engineering, Princeton University, Princeton,
USA, \texttt{kulkarni@princeton.edu.}}
\and 
Eun Jee Lee\thanks{Program in Applied and Computational Mathematics, Princeton University, Princeton,
USA, \texttt{ejlee@princeton.edu.}}
}

\maketitle
\begin{abstract}
This paper develops upper and lower bounds on the influence measure in a network, more precisely, the expected number of nodes that a seed set can influence in the independent cascade model. In particular, our bounds exploit nonbacktracking walks, Fortuin--Kasteleyn--Ginibre (FKG) type inequalities, and are computed by message passing implementation. Nonbacktracking walks have recently allowed for headways in community detection, and this paper shows that their use can also impact the influence computation. Further, we provide a knob to control the trade-off between the efficiency and the accuracy of the bounds. Finally, the tightness of the bounds is illustrated with simulations on various network models.
\end{abstract}

\section{Introduction}
Influence propagation is concerned with the diffusion of information (or viruses) from initially influenced (or infected) nodes, called \emph{seeds}, in a network. 
Understanding how information propagates in networks has become a central problem in a broad range of fields, such as 
viral marketing \cite{leskovec2007dynamics}, sociology \cite{granovetter1978threshold,lopez2008social,watts2002simple}, communication \cite{khelil2002epidemic}, epidemiology \cite{shulgin1998pulse}, and social network analysis \cite{yang2010predicting}.

One of the most fundamental questions on influence propagation is to estimate the \emph{influence}, i.e. the expected number of influenced nodes at the end of the propagation given some seeds.
Estimating the influence is central to various research problems related to influence propagation, such as the widely-known influence maximization problem --- finding a set of $k$ nodes that maximizes the expected number of influenced nodes. 

Recent studies in the influence propagation have proposed heuristic algorithms \cite{kempe2003maximizing,leskovec2007cost,chen2009efficient,goyal2011celf++,wang2012scalable} for the influence maximization problem while using Monte Carlo (MC) simulations to approximate the influence.
Despite its simplicity, approximating the influence via MC simulations is far from ideal for large networks; in particular, MC may require a large amount of computations in order to stabilize the approximation.

To overcome the limitations of Monte Carlo simulations, 
many researchers have been taking theoretical approaches
to approximate the influence of given seeds in a network. 
Draief et al.,~\cite{draief2006thresholds} introduced an upper bound for the influence by using the spectral radius of the adjacency matrix. Tighter upper bounds were later suggested in~\cite{lemonnier2014tight} which relate the ratio of influenced nodes in a network to the spectral radius of the so-called \emph{Hazard} matrix.
Further, improved upper bounds which account for \emph{sensitive} edges were introduced in~\cite{lee2016spectral}.

In contrast, there has been little work on finding a tight lower bound for the influence. 
A few exceptions include the work by Khim et al.~\cite{khim2016computing}, where the lower bound is obtained by only considering the infections through the \emph{maximal-weighted} paths. 

In this paper, we propose both upper and lower bounds on the influence using nonbacktracking walks and Fortuin--Kasteleyn--Ginibre (FKG) type inequalities. 
The bounds can be efficiently obtained by message passing implementation. This shows that nonbacktracking walks can also impact influence propagation, making another case for the use of nonbacktracking walks in graphical model problems as in~\cite{krzakala2013spectral, karrer2014percolation, bordenave2015non, abbe2015detection}, discussed later in the paper. Further, we provide a parametrized version of the bounds that can adjust the trade-off between the efficiency and the accuracy of the bounds.

\section{Background}
We introduce here the independent cascade model and provide background for the main results.  

\begin{definition}[Independent Cascade Model]
Consider a directed graph $G=(V,\vec{E})$ with $|V|=n$, a transmission probability matrix ${\cal P}\in [0,1]^{n\times n}$, and a seed set $S_0\subseteq V$.
For all $u\in V$, let $N(u)$ be the set of neighbors of node $u$.
The independent cascade model $IC(G,{\cal P},S_0)$ sequentially generates the infected set $S_t\subseteq V$ for each discrete time $t\geq 1$ as follows.
At time $t$, $S_t$ is initialized to be an empty set. 
Then, each node $u\in S_{t-1}$ attempts to infect $v\in N(u)\cap (V\setminus \cup_{i=0}^{t-1}S_{i})$ with probability ${\cal P}_{uv}$, i.e. node $u$ infects its uninfected neighbor $v$ with probability ${\cal P}_{uv}$.
If $v$ is infected at time $t$, add $v$ to $S_t$. 
The process stops at $T$ if $S_T=\emptyset$ at the end of the step $t=T$.
The set of the infected nodes at the end of propagation is defined as $S=\cup_{i=0}^{T-1}S_t$.
\end{definition}

We often refer as an edge $(u,v)$ being \emph{open} if node $u$ infects node $v$. 
The IC model can also be defined on an undirected graph by replacing each edge $e=(u,v)$ which has the transmission probability $p_e$ with two directed edges $\overrightarrow{(u,v)}, \overrightarrow{(v,u)}$ and assigning transmission probabilities ${\cal P}_{uv}={\cal P}_{vu}=p_e$.

The IC model is equivalent to the live-arc graph model, where the infection happens at once, rather than sequentially. The live-arc graph model first decides the state of every edge with a Bernoulli trial, i.e. edge $(u,v)$ is open independently with probability ${\cal P}_{uv}$ and closed, otherwise. Then, the set of infected nodes is defined as the nodes that are connected to at least one of the seeds by the open edges. 

\begin{definition}[Influence]
The expected number of nodes that are infected at the end of the propagation process is called the \emph{influence} (rather than the expected influence, with a slight abuse of terminology) of $IC(G,{\cal P},S_0)$, and is defined as
\vspace{-0.2em}
\begin{eqnarray}
\sigma(S_0)&=&\sum_{v\in V}\prob(v \text{ is infected}).
\end{eqnarray}
\end{definition}
\vspace{-0.5em}

It is shown in~\cite{chen2010scalable} that computing the influence $\sigma(S_0)$ in the independent cascade model $IC(G,{\cal P},S_0)$ is $\#$P-hard, even with a single seed, i.e. $|S_0|=1$. 

Next, we define nonbacktracking (NB) walks on a directed graph.
Nonbacktracking walks have already been used for studying the characteristics of networks.
To the best of our knowledge, the use of NB walks in the context of epidemics was first introduced in the paper of Karrer et al.~\cite{karrer2010message} and later applied to percolation in~\cite{karrer2014percolation}. 
In particular, Karrer et al. reformulate the spread of infections as a message passing process and demonstrate how the resulting equations can be used to calculate the upper bound on the number of nodes that are susceptible at a given time. 
As we shall see, we take a different approach to the use of the NB walks, which focuses on the effective contribution of a node in infecting another node and accumulates such contributions to obtain upper and lower bounds.  
More recently, nonbacktracking walks are used for community detection~\cite{krzakala2013spectral, bordenave2015non, abbe2015detection}.

\begin{definition}[Nonbacktracking Walk]
Let $G=(V,E)$ be a directed graph. A nonbacktracking walk of length $k$ is defined as
$w^{(k)}=(v_0,v_1,\dots,v_k)$,
where $v_i\in V$ and $(v_{i-1},v_{i})\in E$ for all $i\in [k]$, and
$v_{i-1}\neq v_{i+1}$ for $i\in [k-1]$.
\end{definition}

We next recall a key inequality introduced by Fortuin et. al~\cite{fortuin1971correlation}. 
\begin{thm}[FKG Inequality]
Let $\Gamma$ be a finite partially ordered set, ordered by $\prec$ with ${(\Gamma, \prec)}$ a distributive lattice and $\mu$ be a positive measure on $\Gamma$ satisfying the following condition: for all $x, y\in \Gamma$,
\begin{eqnarray}
\mu(x\wedge y)\mu(x\vee y)&\geq&\mu(x)\mu(y),\nonumber
\end{eqnarray}
where $x\wedge y=\max\{z\in \Gamma:z\preceq x, z\preceq y\}$  and ${x\vee y=\min\{z\in \Gamma:y\preceq z, y\preceq z\}}$.
Let $f$ and $g$ be both increasing (or both decreasing) functions on $\Gamma$. Then,
\begin{eqnarray}
(\sum_{x\in \Gamma}\mu(x))(\sum_{x\in \Gamma}f(x)g(x)\mu(x))&\geq&(\sum_{x\in \Gamma}f(x)\mu(x))(\sum_{x\in \Gamma}g(x)\mu(x)).
\end{eqnarray}
\end{thm}

FKG inequality is instrumental in studying the influence propagation since the probability that a node is influenced is nondecreasing with respect to the partial order of random variables describing the states, open or closed, of the edges.

\section{Nonbacktracking bounds on the influence}
In this section, we presents upper and lower bounds on the influence in the independent cascade model and explain the motivations and intuitions of the bounds. 
The bounds are computed efficiently by algorithms which utilize nonbacktracking walks and FKG inequalities. 
In particular, the upper bound on a network based on a graph $G(V,E)$ runs in $O(|V|^2+|V||E|)$ and the lower bound runs in $O(|V|+|E|)$, whereas Monte Carlo simulation would require $O(|V|^3+|V|^2|E|)$ computations without knowing the variance of the influence, which is harder to estimate than the influence.
The reason for the large computational complexity of MC is that in order to ensure that the standard error of the estimation does not grow with respect to $|V|$, MC requires $O(|V|^2)$ computations. 
Hence, for large networks, where MC may not be feasible, our algorithms can still provide bounds on the influence. 

Furthermore, from our computable upper $\upinf$ and lower bounds $\lowinf$, we can compute an upper bound on the variance given by $(\upinf-\lowinf)^2/4$.
This could be used to estimate the number of computations needed by MC.
Computing the upper bound on the variance with the proposed bounds can be done in $O(|V|^2+|V||E|)$, whereas computing the variance with MC simulation requires $O(|V|^5+|V|^4|E|)$.

\subsection{Nonbacktracking upper bounds (\nbub)}
We start by defining the following for the independent cascade model $IC(G,{\cal P},S_0)$, where $G=(V,E)$ and $|V|=n$.

\begin{definition}
For any $v\in V$, we define the set of in-neighbors $N^-(v)=\{u\in V: (u,v)\in E\}$ and the set of out-neighbors $N^+(v)=\{u\in V: (v,u)\in E\}$.
\end{definition}

\begin{definition}
For any $v\in V$ and $l\in \{0,\dots,n-1\}$, the set ${P_l(S_0\to v)}$ is defined as the set of all paths with length $l$ from any seed $s\in S_0$ to $v$. We call a path $P$ is \emph{open} if every edge in $P$ is open. 
\end{definition}

\begin{definition} For any $u,v\in V$,\: $l\in \{0,\dots,n\!-\!1\}$,\: and $S\subseteq V$,\: we define the events
\begin{eqnarray}
A(v)&=&\{v \text{ is infected}\}\\
A_l(v)&=&\cup_{P\in P_l(S_0\to v)}\{P\text{ is open}\}\\
A_l(u\to v)&=&\cup_{P\in P_l(S_0\to u), P\notni v}\{P\text{ is open and edge }(u,v)\text{ is open}\}\\
A_{l,S}(v)&=&\cup_{P\in \{P'\in P_l(S_0\to v):P'\notni w, \forall w\in S\}}\{P\text{ is open}\},
\end{eqnarray}
and the probabilities
\begin{eqnarray}
p(v)&=&\prob(A(v))\\
p_l(v)&=&\prob(A_l(v))\\
p_l(u\to v)&=&\prob(A_l(u\to v)).
\end{eqnarray}
\end{definition}
\noindent In other words, $A_l(v)$ is the event that node $v$ is infected by open paths of length $l$, $A_l(u\to v)$ is the event that $v$ is infected by node $u$ with open paths of length $l+1$, i.e. there exists an open path of length $l+1$ from a seed to $v$ that ends with edge $(u,v)$, and $A_{l,S}(v)$ is the event that node $v$ is infected by length $l$ open paths which do not include any node in $S$.

\begin{lemma}\label{lemma1} For any $v\in V$,
\par\vspace{-2.2em}
\begin{eqnarray}
p(v)&\leq&1-\prod_{l=0}^{n-1}(1-p_l(v)).
\end{eqnarray}
For any $v\in V$ and $l\in\{0,\dots,n-1\}$,
\par\vspace{-0.5em}
\begin{eqnarray}
p_l(v)&\leq&1-\prod_{u\in N^-(v)}(1-p_l(u\to v)).
\end{eqnarray}
\end{lemma}
\par\vspace{-0.5em}

\begin{proof}
Recall that $p(v)=\prob(A(v))$, $p_l(v)=\prob(A_l(v))$, and $p_l(u\to v)=\prob(A_l(u\to v))$.
\begin{eqnarray}
p(v)&=&\prob(\cup_{l=0}^{n-1}A_l(v))\\
&=&1-\prob(\cap_{l=0}^{n-1}A_l(v)^C)\\
&\leq&1-\prod_{l=0}^{n-1}\prob(A_l(v)^C)\label{eqn:fkg1}\\
&=&1-\prod_{l=0}^{n-1}(1-p_l(v)).
\end{eqnarray}
Equation (\ref{eqn:fkg1}) follows from positive correlation among the events $A_l(v)^C$ for all $v\in V$, which can be proved by FKG inequality. Similarly,
\begin{eqnarray}
p_l(v)&=&\prob(\cup_{u\in N^-(v)}A_l(u\to v))\\
&=&1-\prob(\cap_{u\in N^-(v)}A_l(u\to v)^C)\\
&\leq&1-\prod_{u\in N^-(v)}\prob(A_l(u\to v)^C)\\
&=&1-\prod_{u\in N^-(v)}(1-p_l(u\to v)).
\end{eqnarray}
\end{proof}

Lemma~\ref{lemma1} suggests that given $p_l(u\to v)$, we may compute an upper bound on the influence. Ideally, $p_l(u\to v)$ can be computed by considering all paths with length $l$. However, this results in exponential complexity $O(n^l)$, as $l$ gets up to $n-1$. Thus, we present an efficient way to compute an upper bound $\ub_l(u\to v)$ on $p_l(u\to v)$, which in turns gives an upper bound $\ub_l(v)$ on $p_l(v)$, with the following recursion formula.

\begin{definition}\label{up:recursion} For all $l\in \{0,\dots,n-1\}$ and $u,v\in V$ such that $(u,v)\in E$, $\ub_l(u)\in [0,1]$ and $\ub_l(u\to v)\in [0,1]$ are defined recursively as follows.
\\\emph{Initial condition:} For every\: $s\!\in\!S_0$,\: $s^+\!\!\in\!N^+(s)$,\: $u\!\in\!V\!\setminus\!S_0$,\: and $v\!\in\!N^+(u)$,
\begin{eqnarray}
&&\ub_0(s)=1,\;\ub_0(s\to s^+)={\cal P}_{ss^+}\\
&&\ub_0(u)=0,\;\ub_0(u\to v)=0.
\end{eqnarray}
\emph{Recursion:} For every\: $l>0$,\: $s\!\in\!S_0$,\: $s^+\!\!\in\!N^+(s)$,\: $s^-\!\!\in\!N^-(s)$,\: $u\!\in\!V\!\setminus\!S_0$,\: and $v\!\in\!N^+(u)\!\setminus\!S_0$,
\begin{eqnarray}
&&\ub_l(s)=0, \ub_l(s\to s^+)=0, \ub_l(s^-\to s)=0\label{eqn:seed}\\
&&\ub_l(u)=1-\!\!\!\!\!\prod_{w\in N^-(u)}\!\!\!\!\!(1-\ub_{l-1}(w\to u))\label{eqn:inprob}\\
&&\ub_l(u\to v)={\begin{cases} 
      {\cal P}_{uv}(1-\frac{1-\ub_l(u)}{1-\ub_{l-1}(v\to u)}), & \text{if } v\!\in\!N^-(u) \\          {\cal P}_{uv}\ub_l(u), & \text{otherwise. }\end{cases}}\label{eqn:nb}
\end{eqnarray}
\end{definition}
Equation (\ref{eqn:seed}) follows from that for any seed node $s\in S_0$ and for all $l>0$, $p_l(s)=0$, $p_l(s\to s^+)=0$, $p_l(s^-\to s)=0$.
A naive way to compute $\ub_l(u\to v)$ is $\ub_l(u\to v)={\cal P}_{uv}\ub_{l-1}(u)$, but this results in an extremely loose bound due to the backtracking. For a tighter bound, we use nonbacktracking in Equation (\ref{eqn:nb}), i.e. when computing $\ub_l(u\to v)$, we ignore the contribution of $\ub_{l-1}(v\to u)$.

\begin{thm}
For any independent cascade model $IC(G,{\cal P},S_0)$,
\begin{eqnarray}
\sigma(S_0)&\leq&\sum_{v\in V}(1-\prod_{l=0}^{n-1}(1-\ub_l(v)))\eqqcolon\upinf(S_0),\label{eqn:upresult}
\end{eqnarray}
where $\ub_l(v)$ is obtained recursively as in Definition~\ref{up:recursion}.
\end{thm}

\begin{proof}
We provide a proof by induction.
The initial condition, for $l=0$, can be easily checked. For every\: $s\!\in\!S_0$,\: $s^+\!\!\in\!N^+(s)$,\: $u\!\in\!V\!\setminus\!S_0$,\: and $v\!\in\!N^+(u)$,
\begin{eqnarray}
p_0(s)=1&\leq&\ub_0(s)=1\\
p_0(s\to s^+)={\cal P}_{ss^+}&\leq&\ub_0(s\to s^+)={\cal P}_{ss^+}\\
p_0(u)=0&\leq&\ub_0(u)=0\\
p_0(u\to v)=0&\leq&\ub_0(u\to v)=0.
\end{eqnarray}
For each $l\leq L$, assume that $p_l(v)\leq \ub_l(v)$ and $p_l(u\to v)\leq \ub_l(u\to v)$ for all $u,v\in V$. 

Since $p_l(s)=p_l(s\to s^+)=p_l(s^-\to s)=0$ for every $l\geq 1$,\: $s\!\in\!S_0$,\: $s^+\!\!\in\!N^+(s)$,\: and $s^-\!\!\in\!N^-(s)$,
it is sufficient to show $p_{L+1}(v)\leq \ub_{L+1}(v)$ and $p_{L+1}(u\to v)\leq \ub_{L+1}(u\to v)$ for all $u\!\in\!V\!\setminus\!S_0$,\: and $v\!\in\!N^+(u)$.

For simplicity, for any pair of events $(A,B)$, use the notation $AB$ for $A\cap B$. 

For any $v\in V\setminus S_0$,
\begin{eqnarray}
p_{L+1}(v)=\prob(\cup_{u\in N^-(v)} E_{uv}A_{L,\{v\}}(u)),
\end{eqnarray}
where $E_{uv}$ denotes the event that edge $(u,v)$ is open, i.e. $\prob(E_{uv})={\cal P}_{uv}$. Thus,
\begin{eqnarray}
p_{L+1}(v)&=&1-\prob(\cap_{u\in N^-(v)}(E_{uv}A_{L,\{v\}}(u))^C)\\
&\leq&1-\prod_{u\in N^-(v)}(1-\prob(E_{uv}A_{L,\{v\}}(u)))\label{eqn:upfkg}\\
&=&1-\prod_{u\in N^-(v)}(1-p_{L}(u\to v))\\
&\leq&1-\prod_{u\in N^-(v)}(1-\ub_{L}(u\to v))=\ub_{L+1}(v),\label{eqn:upassume}
\end{eqnarray}
where Equation (\ref{eqn:upfkg}) is obtained by the positive correlation among the events $E_{uv}A_{L,\{v\}}(u)$, and Equation (\ref{eqn:upassume}) comes from the assumption. 

For any $v\in V\setminus S_0$ and $w\in N^+(v)$, 
\begin{eqnarray}
p_{L+1}(v\to w)&=&\prob(E_{vw}A_{L+1,\{w\}}(v)).\\
&=&{\cal P}_{vw}\prob(A_{L+1,\{w\}}(v))\label{eqn:upind}
\end{eqnarray}
Equation (\ref{eqn:upind}) follows from the independence between the events $E_{vw}$ and $A_{L+1,\{w\}}(v)$.\\
If $w\in N^-(v)$,
\begin{eqnarray}
p_{L+1}(v\to w)&=&{\cal P}_{vw}\prob(\cup_{u\in N^-(v)\setminus \{w\}}E_{uv}A_{L,\{v,w\}}(u))\\
&\leq&{\cal P}_{vw}\left(1-\!\!\!\!\!\!\prod_{u\in N^-(v)\setminus \{w\}}\!\!\!\!\!\!\!\!(1-\prob(E_{uv}A_{L,\{v,w\}}(u)))\right)\\
&\leq&{\cal P}_{vw}\left(1-\!\!\!\!\!\!\prod_{u\in N^-(v)\setminus \{w\}}\!\!\!\!\!\!\!\!(1-p_L(u\to v))\right)\label{eqn:uputov}\\
&\leq&{\cal P}_{vw}\left(1-\!\!\!\!\!\!\prod_{u\in N^-(v)\setminus \{w\}}\!\!\!\!\!\!\!\!(1-\ub_L(u\to v))\right)\label{eqn:uputov},
\end{eqnarray}
Equation (\ref{eqn:uputov}) holds, since the two events satisfy $E_{uv}A_{L,\{v,w\}}(u)\subseteq E_{uv}A_{L,\{v\}}(u)$.\\
Recall that, if $w\in N^-(v)$,
\begin{eqnarray}
\ub_{L+1}(v\to w)&=&{\cal P}_{vw}(1-\frac{1-\ub_{L+1}(v)}{1-\ub_{L}(w\to v)})\\
&=&{\cal P}_{vw}(1-\prod_{u\in N^-(v)\setminus \{w\}}(1-\ub_L(u\to v))).
\end{eqnarray}
Thus, $p_{L+1}(v\to w)\leq\ub_{L+1}(v\to w)$, for all $w\in N^+(v)\cap N^-(v)$.\\
If $w\notin N^-(v)$,
\begin{eqnarray}
p_{L+1}(v\to w)&=&{\cal P}_{vw}\prob(\cup_{u\in N^-(v)}E_{uv}A_{L,\{v,w\}}(u))\\
&\leq&{\cal P}_{vw}\left(1-\!\!\!\!\!\!\prod_{u\in N^-(v)}\!\!\!\!(1-\ub_L(u\to v))\right)\\
&=&{\cal P}_{vw}\ub_{L+1}(v)=\ub_{L+1}(v\to w),
\end{eqnarray}
Hence, $p_{L+1}(v\to w)\leq\ub_{L+1}(v\to w)$, for all $w\in N^+(v)$, concluding the induction proof.\\
Finally, by Lemma~\ref{lemma1},
\begin{eqnarray}
\sigma(S_0)&\leq&\sum_{v\in V}(1-\prod_{l=0}^{n-1}(1-p_l(v)))\\
&\leq&\sum_{v\in V}(1-\prod_{l=0}^{n-1}(1-\ub_l(v)))=\upinf(S_0).
\end{eqnarray}
\end{proof}

\subsection{Algorithm for NB-UB}
Next, we present Nonbacktracking Upper Bound (NB-UB) algorithm which efficiently computes $\ub_l(v)$ and $\ub_l(u\to v)$ by message passing. At $l$-th iteration, the variables in \nbub\ represent as follows.
\begin{itemize}[leftmargin=0.5cm,topsep=-2pt,parsep=0pt] 
    \item[$\cdot$] $S_l$ : set of nodes that are processed at $l$-th iteration.
    \item[$\cdot$] $\incur(v) = \{(u, \ub_{l-1}(u\to v)): u \text{ is an in-neighbor of }v, \text{ and } u\in S_{l-1}\}$,\\
    set of pairs (previously processed in-neighbor $u$ of $v$, incoming message from $u$ to $v$).
    \item[$\cdot$] $\vcur(v)=\{u:u \text{ is a in-neighbor of }v, \text{ and } u\in S_{l-1}\}$,\\
    set of in-neighbor nodes of $v$ that were processed at the previous step.
    \item[$\cdot$] $\incur(v)[u]=\ub_{l-1}(u\to v)$, the incoming message from $u$ to $v$.
    \item[$\cdot$] $\innext(v) = \{(u, \ub_{l}(u\to v)): u \text{ is an in-neighbor of }v, \text{ and } u\in S_{l}\}$,\\
    set of pairs (currently processed in-neighbor $u$, next iteration's incoming message from $u$ to $v$).
\end{itemize}

\begin{algorithm}[h!]
   \caption{Nonbacktracking Upper Bound (NB-UB)}
   \label{alg:upper}
\begin{algorithmic}
   \State {\bfseries Initialize:} $\ub_l(v) = 0$ for all $0\leq l\leq n-1$ and $v\in V$
   \State {\bfseries Initialize:} Insert $(v, 1)$ to $\innext(v)$ for all $v\in S_0$
   \For{$l=0$ {\bfseries to} $n-1$}
   \For{$u \in S_l$}
   \State $\incur(u) = \innext(u)$
   \State Clear $\innext(u)$
   \State $\ub_l(u) = \inprob(\incur(u))$
   \For{$v\in N^+(u)\setminus S_0$}
   \State $S_{l+1}.\minsert(v)$
   \If{$v\in \vcur(u)$}
   \State $\ub_l(u\to v)=\outprob(\incur(u)[v],\ub_l(u),{\cal P}_{uv})$
   \State $\innext(v).\minsert((u,\ub_l(u\to v)))$.
   \Else 
   \State $\ub_l(u\to v)=\outprob(0,\ub_l(u),{\cal P}_{uv})$
   \State $\innext(v).\minsert((u,\ub_l(u\to v)))$.
   \EndIf
   \EndFor
   \EndFor
   \EndFor
   \State {\bfseries Output:} $\ub_l(u)$ for all $l$, $u$
\end{algorithmic}
\end{algorithm}

\vspace{1em}\noindent At the beginning, every seed node $v\in S_0$ is initialized such that $\incur(v)=\{(v,1)\}$. For each $l$-th iteration, every node $u$ in $S_l$ is processed as follows. First, $\inprob(\incur(u))$ computes $\ub_l(u)$ as in Equation (\ref{eqn:inprob}). Second, $u$ passes a message to its neighbor $v\in N^+(u)\setminus S_0$ along the edge $(u,v)$, and $v$ stores (inserts) the message in $\innext(v)$ for the next iteration. The message contains 1) the source of the message, $u$, and 2) $\ub_l(u\to v)$, which is computed as in Equation (\ref{eqn:nb}), by the function $\outprob$. Finally, the algorithm outputs $\ub_l(u)$ for all $u\in V$ and $l\!\in\!\{0,\dots,n\!-\!1\}$, and the upper bound $\upinf(S_0)$ is computed by Equation (\ref{eqn:upresult}). 

\begin{figure*}[h]
\centering
\includegraphics[width=\textwidth]{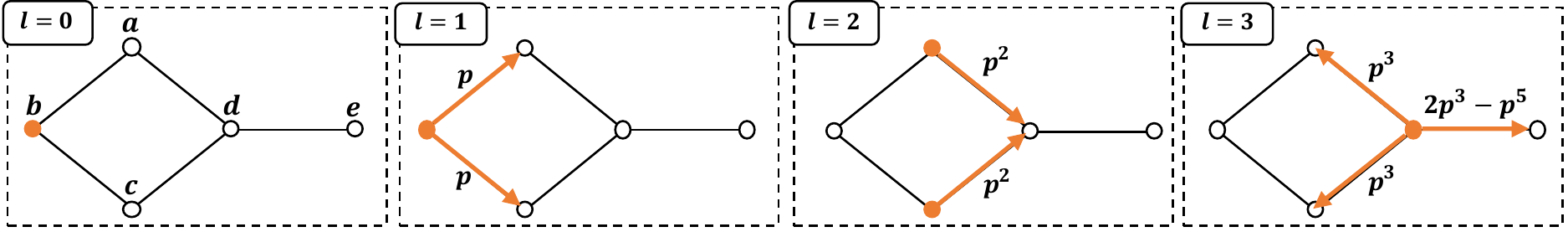}
\caption{The step-wise illustration of \nbub\ algorithm on the example network.}
\label{fig:upper}
\end{figure*}

Next, we illustrate how the algorithm runs on a small network in Figure~\ref{fig:upper}.
The independent cascade model $IC(G,{\cal P},S_0)$ is defined on an undirected graph $G=(V,E)$, where $V=\{a,b,c,d\}$, $S_0=\{b\}$, and every edge has the same transmission probability $p$. For each $l$, Table~\ref{table:up_ex} shows the values of the key variables, $S_l$, $\incur$, and $\ub_l$, in the algorithm and $\ub_l(u\to v)$ for every pair $u,v$ such that $u\in S_l$ and $v\in N^+(u)\setminus S_0$.

\begin{table}[h]
\centering
{\footnotesize
\begin{tabular*}{\textwidth}{c|c|c|c|c|c|c|c|c}
& \multicolumn{2}{c|}{$l=0$} & \multicolumn{2}{c|}{$l=1$} & \multicolumn{2}{c|}{$l=2$} &  \multicolumn{2}{c}{$l=3$}\\ \hline
$S_l$ & \multicolumn{2}{c|}{$\{b\}$} & \multicolumn{2}{c|}{$\{a,c\}$} & \multicolumn{2}{c|}{$\{d\}$} & \multicolumn{2}{c}{$\{a,c,e\}$}\\ \hline
& $\incur$ & $\ub_0$ & $\incur$ & $\ub_1$ & $\incur$ & $\ub_2$ & $\incur$ & $\ub_3$\\ \hline
$a$ & $\varnothing$ & 0 & $\{(b,p)\}$ & $p$ & $\varnothing$ & 0 & $\{(d,p^3)\}$ & $p^3$\\ \hline
$b$ & $\{(b,1)\}$ & 1 & $\varnothing$ & 0 & $\varnothing$ & 0 & $\varnothing$ & 0\\ \hline
$c$ & $\varnothing$ & 0 & $\{(b,p)\}$ & $p$ & $\varnothing$ & 0 & $\{(d,p^3)\}$ & $p^3$\\ \hline
$d$ & $\varnothing$ & 0 & $\varnothing$ & 0 & {\begin{tabular}[c]{@{}c@{}}$\{(a,p^2),$\\ $(c,p^2)\}$\end{tabular}} & $2p^2-p^4$ & $\varnothing$ & 0\\ \hline
$e$ & $\varnothing$ & 0 & $\varnothing$ & 0 & $\varnothing$ & 0 & $\{(d,2p^3-p^5)\}$ & $2p^3-p^5$\\ \hline
{\begin{tabular}[c]{@{}c@{}}Out\\-Prob\end{tabular}} & \multicolumn{2}{c|}{\begin{tabular}[c]{@{}c@{}}$\ub_0(b\to a)=p$\\ $\ub_0(b\to c)=p$\end{tabular}} & \multicolumn{2}{c|}{\begin{tabular}[c]{@{}c@{}}$\ub_1(a\to d)=p^2$\\ $\ub_1(c\to d)=p^2$\end{tabular}} & \multicolumn{2}{c|}{\begin{tabular}[c]{@{}c@{}}$\ub_2(d\to a)=p^3$\\ $\ub_2(d\to c)=p^3$\\ $\ub_2(d\to e)=2p^3-p^5$\end{tabular}} & \multicolumn{2}{c}{\begin{tabular}[c]{@{}c@{}}$\ub_3(a\to d)=0$\\ $\ub_3(c\to d)=0$\\ $\ub_3(e\to d)=0$\end{tabular}}
\end{tabular*}
}
\caption{The values of the key variables in \nbub\ algorithm on the example network in Figure~\ref{fig:upper}.}
\label{table:up_ex}
\end{table}

For example, at $l=2$, since $S_2=\{d\}$, node $d$ is processed. Recall that, at $l=1$, node $a$ sent the message $(a,\ub_1(a\to d))$ to $d$, and node $c$ sent the message $(c,\ub_1(c\to d))$ to $d$. Thus,
\begin{eqnarray}
    \incur(d)&=&\{(a,\ub_1(a\to d)),(c,\ub_1(c\to d))\}=\{(a,p^2),(c,p^2)\}\\
    \vcur(d)&=&\{a,c\},
\end{eqnarray}
and node $d$ is processed as follows.\\
First, compute $\ub_2(d)$ as
\begin{eqnarray}
    \ub_2(d)&=&\inprob(\incur(d))\\
    &=&1-(1-\ub_1(a\to d))(1-\ub_1(c\to d))=2p^2-p^4.
\end{eqnarray}
Next, set $S_{3}=N^+(d)\!\setminus\!S_0=\{a,c,e\}$, and compute $\ub_2(d\to a)$, $\ub_2(d\to c)$, and $\ub_2(d\to e)$ as 
\begin{eqnarray}
\ub_2(d\to a)&=&\outprob(\ub_1(a\to d),\ub_2(d),{\cal P}_{da})\\
&=&{\cal P}_{da}(1-\frac{1-\ub_2(d)}{1-\ub_1(a\to d)})=p^3\\
\ub_2(d\to c)&=&\outprob(\ub_1(c\to d),\ub_2(d),{\cal P}_{dc})\\
&=&{\cal P}_{dc}(1-\frac{1-\ub_2(d)}{1-\ub_1(c\to d)})=p^3\\
\ub_2(d\to e)&=&\outprob(0,\ub_2(d),{\cal P}_{de})\\
&=&{\cal P}_{de}(1-\frac{1-\ub_2(d)}{1-0})={\cal P}_{de}\ub_2(d)=2p^3-p^5.
\end{eqnarray}
Then, node $d$ send messages $(d,\ub_2(d\to a))$ to $a$, $(d,\ub_2(d\to c))$ to $c$, and $(d,\ub_2(d\to e))$ to $e$, concluding the process of the $l=2$ step.\\

\noindent {\bf Computational complexity}: Notice that for each iteration $l\!\in\!\{0,\dots,n\!-\!1\}$, the algorithm accesses at most $n$ nodes, and for each node, the functions $\inprob$ and $\outprob$ are computed in $O(\deg(v))$. Therefore, the worst case computational complexity is $O(|V|^2+|V||E|)$.

\subsection{Nonbacktracking lower bounds (\nblb)}
A naive way to compute a lower bound on the influence in a network $IC(G,{\cal P},S_0)$ is to reduce the network to a (spanning) tree network, by removing edges. Then, since there is a unique path from a node to another, we can compute the influence of the tree network, which is a lower bound on the influence in the original network, in $O(|V|)$. We take this approach of generating a subnetwork from the original network, yet we avoid the significant gap between the bound and the influence by considering the following directed acyclic subnetworks, in which there is no backtracking walk. 

\begin{definition}[Min-distance Directed Acyclic Subnetwork] Consider an independent cascade model $IC(G,{\cal P},S_0)$ with $G = (V,E)$ and $|V|=n$. Let $d(S_0,v)$ be the minimum distance from a seed in $S_0$ to $v$. A minimum-distance directed acyclic subnetwork (MDAS), $IC(G',{\cal P}',S_0)$, where $G'=(V',E')$, is obtained as follows.
\begin{itemize}[leftmargin=0.5cm,topsep=-2pt,parsep=0pt] 
    \item[$\cdot$] $V' = \{v_1, ... , v_n\}$ is an ordered set of nodes such that $d(S_0,v_i)\leq d(S_0,v_j)$, for every $i < j$.
    \item[$\cdot$] $E'=\{(v_i,v_j)\in E:i<j\}$, i.e. remove edges from $E$ whose source node comes later in the order than its destination node to obtain $E'$.
    \item[$\cdot$] ${\cal P}'_{v_iv_j}={\cal P}_{v_iv_j}$, if $(v_i,v_j)\in E'$, and ${\cal P}'_{v_iv_j}=0$, otherwise.
\end{itemize}
\end{definition}
\noindent If there are multiple ordered sets of vertices that satisfying the condition, we may choose an order arbitrarily.

For any $k\in [n]$, let $p(v_k)$ be the probability that $v_k\in V'$ is infected in the MDAS, $IC(G',{\cal P}',S_0)$. Since $p(v_k)$ is equivalent to the probability of the union of the events that an in-neighbor $u_i\in N^-(v_k)$ infects $v_k$,\: $p(v_k)$ can be computed by the principle of inclusion and exclusion. Thus, we may compute a lower bound on $p(v_k)$, using Bonferroni inequalities, if we know the probabilities that in-neighbors $u$ and $v$ both infects $v_k$, for every pair $u,v\in N^-(v_k)$. However, computing such probabilities can take $O(k^k)$. Hence, we present $\lb_l(v_k)$ which efficiently computes a lower bound on $p(v_k)$ by the following recursion.

\begin{definition}\label{def:low}For all $v_k\in V'$, $\lb(v_k)\in [0,1]$ is defined by the recursion on $k$ as follows.\\
\emph{Initial condition:} for every $v_s\in S_0$,
\begin{eqnarray}
\lb(v_s)=1.
\end{eqnarray}
\emph{Recursion:} for every $v_k\in V'\setminus S_0$,
\begin{eqnarray}
\lb(v_k)=\sum_{i=1}^{m^*}\left({\cal P}'_{u_iv_k}\lb(u_i)(1-\sum_{j=1}^{i-1}{\cal P}'_{u_jv_k})\right),\label{eqn:low1}
\end{eqnarray}
where $N^-(v_k)\!=\!\{u_1,\dots,u_m\}$ is the ordered set of in-neighbors of $v_k$ in $IC(G',{\cal P}',S_0)$ and ${m^*\!=\!\max\{m'\leq m:\sum_{j=1}^{m'-1}{\cal P}'_{u_jv_k}\leq 1\}}$.
\end{definition}
\noindent{\bf Remark.} Since the $i$-th summand in Equation (\ref{eqn:low1}) can utilize $\sum_{j=1}^{i-2}{\cal P}'_{u_jv_k}$, which is already computed in $(i\!-\!1)$-th summand, to compute $\sum_{j=1}^{i-1}{\cal P}'_{u_jv_k}$, the summation takes at most $O(\deg(v_k))$.

\begin{thm}
For any independent cascade model $IC(G,{\cal P},S_0)$ and its MDAS $IC(G',{\cal P}',S_0)$, 
\begin{eqnarray}
\sigma(S_0)\geq\sum_{v_k\in V'}\lb(v_k)\eqqcolon\lowinf(S_0),
\end{eqnarray}
where $\lb(v_k)$ is obtained recursively as in Definition~\ref{def:low}.
\end{thm}

\begin{proof}
We provide a proof by induction. For any $v_k\in V'$, let $A(v_k)$ be the event that node $v_k$ is infected in MDAS $IC(G',{\cal P}',S_0)$, and for every edge $(v_j,v_k)$, let $E_{v_j,v_k}$ be the event that edge $(v_j,v_k)$ is open, i.e. $\prob(E_{v_j,v_k})={\cal P}'_{v_jv_k}$. Recall that $p(v_k)=\prob(A(v_k))$. 

The initial condition $k=1$ holds, since $p(v_1)=1\geq \lb(v_1)=1$ ($v_1$ is a seed). 

For every $k\leq K$, assume $p(v_k)\geq \lb(v_k)$.

For the node $v_{K+1}$,
\begin{eqnarray}
p(v_{K+1})&=&\prob(\cup_{v_j\in N^-(v_{K+1})}E_{v_jv_{K+1}}A(v_j)).
\end{eqnarray}
We re-label vertices in $N^-(v_{K+1})=\{u_1,\dots,u_{m(K+1)}\}$ where $m(K+1)=\mathrm{in}\text{-}\!\deg(v_{K+1})$, and let ${\cal Q}_{iK+1}={\cal P}'_{u_iv_{K+1}}$. Then, for any integer $m\leq m(K+1)$,
\begin{eqnarray}
p(v_{K+1})&=&\prob(\cup_{i=1}^{m(K+1)}E_{u_iv_{K+1}}A(u_i))\\
&\geq&\prob(\cup_{i=1}^{m}E_{u_iv_{K+1}}A(u_i))\\
&\geq&\sum_{i=1}^{m}\prob(E_{u_iv_{K+1}}A(u_i))-\sum_{i=1}^{m}\sum_{j=1}^{i-1}\prob(E_{u_iv_{K+1}}A(u_i)E_{u_jv_{K+1}}A(u_j))\label{lower:pie}\\
&=&\sum_{i=1}^{m}{\cal Q}_{iK+1}\prob(A(u_i))-\sum_{i=1}^{m}\sum_{j=1}^{i-1}{\cal Q}_{iK+1}{\cal Q}_{jK+1}\prob(A(u_i)A(u_j))\label{lower:ind}\\
&\geq&\sum_{i=1}^{m}{\cal Q}_{iK+1}\prob(A(u_i))(1-\sum_{j=1}^{i-1}{\cal Q}_{jK+1})\label{lower:minusbig}.
\end{eqnarray}
Equation (\ref{lower:pie}) follows from the principle of inclusion and exclusion. Equation (\ref{lower:ind}) results from the Independence between the event that an edge ending with $v_{K+1}$ is open and the event that a node $v_i$ is infected where $i<K+1$. Equation (\ref{lower:minusbig}) holds since $\prob(A(u_i))\geq \prob(A(u_i)A(u_j))$. 

Now, define $m^*=\max\{m'\leq m(K+1):\sum_{j=1}^{m'-1}{\cal Q}_{jK+1}\leq 1\}$.
Then,
\begin{eqnarray}
p(v_{K+1})&\geq&\sum_{i=1}^{m^*}{\cal Q}_{iK+1}\prob(A(u_i))(1-\sum_{j=1}^{i-1}{\cal Q}_{jK+1})\\
&\geq&\sum_{i=1}^{m^*}{\cal Q}_{iK+1}\lb(u_i)(1-\sum_{j=1}^{i-1}{\cal Q}_{jK+1})\label{lower:nonneg}\\
&=&\lb(v_{K+1}).
\end{eqnarray}
Equation (\ref{lower:nonneg}) follows since $1-\sum_{j=1}^{i-1}{\cal Q}_{jK+1}\geq0$ for all $i\leq m^*$ by the definition of $m^*$. Thus, $p(v_i)\geq \lb(v_i)$ for all $v_i\in V'$, concluding the induction proof. \\
Finally,
\begin{eqnarray}
\sigma(S_0)&\geq&\sum_{i=1}^{n}p(v_i)\label{eqn:subnetwork}\\
&\geq&\sum_{i=1}^{n}\lb(v_i)=\lowinf(S_0).
\end{eqnarray}
Equation (\ref{eqn:subnetwork}) holds since its right hand side equals to the influence of the MDAS, $IC(G',{\cal P}',S_0)$.
\end{proof}

\subsection{Algorithm for NB-LB}
Next, we present Nonbacktracking Lower Bound (NB-LB) algorithm which efficiently computes $\lb(v_k)$. At $k$-th iteration, the variables in \nblb\ represent as follow.
\begin{itemize}[leftmargin=0.5cm,topsep=-2pt,parsep=0pt]
    \item[$\cdot$] $\lowin(v_k)=\{(\lb(v_t),{\cal P}'_{v_tv_k}):v_t\text{ is an in-neighbor of }v_k\}$, set of pairs (incoming message from an in-neighbor $v_t$ to $v_k$, the transmission probability of edge $(v_t,v_k)$).
    \item[$\cdot$] $\lowin(v_k)_i=(\lowin(v_k)_{i,1},\lowin(v_k)_{i,2})$, the $i$-th pair in $\lowin(v_k)$, for $i\geq 1$.
\end{itemize}

\begin{algorithm}[h]
   \caption{Nonbacktracking Lower Bound (NB-LB)}
   \label{alg:lower}
\begin{algorithmic}
   \State {\bfseries Input:} directed acyclic network $IC(G',{\cal P}',S_0)$
   \State {\bfseries Initialize:} $\lowinf = 0$
   \State {\bfseries Initialize:} Insert $(1,1)$ to $\lowin(v_i)$ for all $v_i\in S_0$
   \For{$k=1$ {\bfseries to} $n$}
   \State $\lb(v_k)=\lowprob(\lowin(v_k))$
   \State $\lowinf \pluseq \lb(v_k)$
   \For{$v_t\in N^+(v_k)\setminus S_0$}
   \State $\lowin(v_t).\minsert((\lb(v_k),{\cal P}'_{v_kv_t}))$
   \EndFor
   \EndFor
   \State {\bfseries Output:} $\lowinf$
\end{algorithmic}
\end{algorithm}
\setlength{\intextsep}{7pt}
 
At the beginning, every seed node $v\in S_0$ is initialized such that $\lowin(v)=\{(1,1)\}$, and $\lowinf=0$. For each $k$-th iteration, node $v_k$ is processed as follows. First, $\lb(v_k)$ is computed as in the Equation (\ref{eqn:low1}), by the  function $\lowprob$, and added to $\lowinf$. Second, $v_k$ passes the message $(\lb(v_k),{\cal P}'_{v_kv_t})$ to its neighbor $v_t\in N^+(v_k)\!\setminus\!S_0$, and $v_t$ stores (inserts) it in $\lowin(v_t)$. Finally, the algorithm outputs $\lowinf$, the lower bound on the influence.

\begin{figure*}[h!t]
\centering
\includegraphics[width=\textwidth]{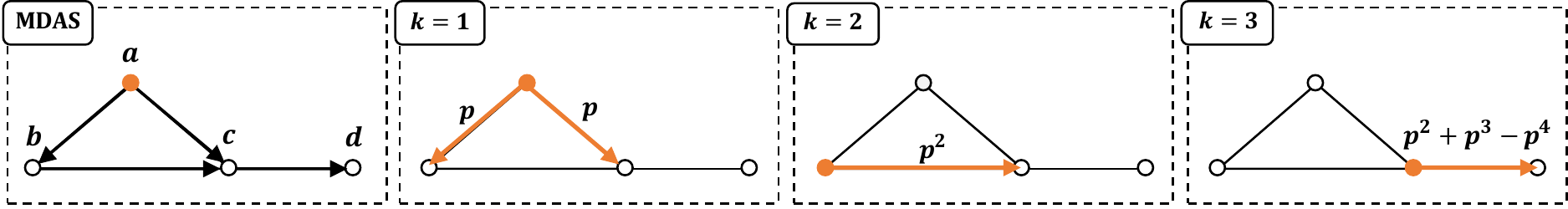}
\caption{The step-wise illustration of \nblb\ on the example network.}
\label{fig:lower}
\end{figure*}

In Figure~\ref{fig:lower}, we show an example for the lower bound computation by \nblb\ on a small network $IC(G,{\cal P},S_0)$ defined on an undirected graph $G=(V,E)$, where $V=\{a,b,c,d\}$, $S_0=\{a\}$, and every edge has the same transmission probability $p$. For each $k$, Table~\ref{table:low_ex} shows the values of the key variables, $\lowin(v_k)$, $\lb(v_k)$, and $(\lb(v_k),{\cal P}'_{v_kv_t})$ for the out-neighbors $v_t\in N^+(v_k)\setminus S_0$, and shows the change in $\lowinf$.

\begin{table}[h]
\centering
{\footnotesize
\begin{tabular}{c|c|c|c|c}
& $k=1$ & $k=2$ & $k=3$ & $k=4$\\ \hline
$v_k$ & $a$ & $b$ & $c$ & $d$\\ \hline
$\lowin(v_k)$ & $\{(1,1)\}$ & $\{(1,p)\}$ & $\{(1,p),(p,p)\}$ & $\varnothing$ \\ \hline
$\lb(v_k)$ & $1$ & $p$ & $p+p^2-p^3$ & $p^2+p^3-p^4$\\ \hline
$N^+(v_k)\setminus S_0$ & $\{b,c\}$ & $\{c\}$ & $\{d\}$ & $\varnothing$ \\ \hline
$(\lb(v_k),{\cal P}'_{v_kv_t})$ to $v_t$ & $(1,p)$ to $b$ and $c$ & $(p,p)$ to $c$ & $(p+p^2-p^3,p)$ to $d$ & \\ \hline
$\lowinf$ & $1$ & $1+p$ & $1+2p+p^2-p^3$ & $1+2p+2p^2-p^4$   
\end{tabular}}
\caption{The values of the key variables in \nblb\ on the example network in Figure~\ref{fig:lower}.}
\label{table:low_ex}
\end{table}

From the undirected network, we obtain MDAS in Figure~\ref{fig:lower} as follows. Since $d(a,S_0)=0$, $d(b,S_0)=d(c,S_0)=1$ and $d(d,S_0)=2$, we order the vertices as $\{v_1\!=\!a,v_2\!=\!b,v_3\!=\!c,v_4\!=\!d\}$ to satisfy that $d(v_i,S_0)\leq d(v_j,S_0)$, for every $i<j$. Then, \nblb\ algorithm process the nodes $\{v_1\!=\!a,v_2\!=\!b,v_3\!=\!c,v_4\!=\!d\}$ sequentially. 

For example, at $k\!=\!3$, node $c$ is processed. Recall that at $k=1$, node $a$ sent the message $(\lb(a),{\cal P}'_{ac})$ to node $c$; at $k=2$, node $b$ sent the message $(\lb(b),{\cal P}'_{bc})$ to node $c$. Thus,
\begin{equation}
    \lowin(c)=\{(\lb(a),{\cal P}'_{ac}),(\lb(b),{\cal P}'_{bc})\}=\{(1,p),(p,p)\}.
\end{equation}
Then, it computes $\lb(c)$ by the function $\lowprob$.
\begin{eqnarray}
    \lb(c)&=&\lowprob(\lowin(c))\\
    &=&{\cal P}'_{ac}\lb(a)+{\cal P}'_{bc}\lb(b)(1-{\cal P}'_{ac})=p+p^2-p^3.
\end{eqnarray}
Recall that $\lowinf=1+p$, at the end of iteration $k=2$, so
\begin{equation}
    \lowinf=1+p+\lb(c)=1+2p+p^2-p^3.
\end{equation}
Next, since $N^+(c)\setminus S_0=\{d\}$, node $c$ sends the message $(\lb(c),{\cal P}_{cd})=(p+p^2-p^3,p)$ to node $d$, concluding the process of the $k=3$ step.\\

\noindent {\bf Computational complexity:} Obtaining an arbitrary directed acyclic subnetwork from the original network takes $O(|V|+|E|)$. Next, the algorithm iterates through the nodes $V'\!=\!\{v_1,\dots,v_n\}$. For each node $v_k$, $\lowprob$ takes $O(\deg(v_k))$ and $v_k$ sends messages to its out-neighbors in $O(\deg(v_k))$. Hence, the worst case computational complexity is $O(|V|+|E|)$.

\subsection{Tunable nonbacktracking bounds}
In this section, we introduce the parametrized version of \nbub\ and \nblb\ which provide control to adjust the trade-off between the efficiency and the accuracy of the bounds. \\

\noindent {\bf Tunable nonbacktracking upper bounds (t\nbub):} The algorithm inputs the parameter $t$, which indicates the maximum length of the paths that the algorithm considers to compute the exact, rather than the upper bound on, probability of infection. For every node $u\in V$, the algorithm computes $p_{{\leq}t}(u)$ that node $u$ is infected by open paths whose length is less than or equal to $t$.
\begin{eqnarray}
p_{{\leq}t}(u)&=&\prob(\cup_{P\in \{\cup_{i=0}^{t}P_i(S_0\to u)\}}\{P\text{ is open}\}).
\end{eqnarray}
Then, we start (non-parameterized) \nbub\ algorithm from $l=t+1$ with the new initial conditions: for all $u\in V$ and $v\in N^+(u)$, 
\begin{eqnarray}
&&\ub_t(u)=p_{\leq t}(u)\\
&&\ub_t(u\to v)=p_t(u\to v)
\end{eqnarray}
Finally, the upper bound by tNB-UB is computed as $\sum_{v\in V}(1-\prod_{l=t}^{n-1}(1-\ub_l(v)))$. 

For higher values of $t$, the algorithm results in tighter upper bounds, while the computational complexity may increase exponentially for dense networks. Thus, this method is most applicable in sparse networks, where the degree of each node is bounded.  

We present here the parametrized algorithms for \nbub. 

\begin{algorithm}[H]
   \caption{Tunable NB-UB (tNB-UB)}
   \label{alg:upperknob}
\begin{algorithmic}
   \State {\bfseries parameter:} non-negative integer $t\leq n-1$
   \State {\bfseries Initialize:} $\ub_t(v) = 0$ for all $t\leq l\leq n-1$ and $v\in V$
   \For{$u\in V$}
   \State $\ub_t(u) =p_{\leq t}(u)$
   \For{$v\in N^+(u)\setminus S_0$}
   \If{$p_t(u\to v)>0$}
   \State $S_{t+1}.\minsert(v)$
   \State $\innext(v).\minsert(u,p_t(u\to v))$
   \EndIf
   \EndFor
   \EndFor
   \For{$l=t+1$ {\bfseries to} $n-1$}
   \For{$u \in S_l$}
   \State $\incur(u) = \innext(u)$
   \State Clear $\innext(u)$
   \State $\ub_l(u) = \inprob(\incur(u))$
   \For{$v\in N^+(u)\setminus S_0$}
   \State $S_{l+1}.\minsert(v)$
   \If{$v\in \incur(u)$}
   \State $\ub_l(u\to v)=\outprob(\incur(u)[v],\ub_l(u),{\cal P}_{uv})$
   \State $\innext(v).\minsert((u,\ub_l(u\to v)))$.
   \Else 
   \State $\ub_l(u\to v)=\outprob(0,\ub_l(u),{\cal P}_{uv})$
   \State $\innext(v).\minsert((u,\ub_l(u\to v)))$.
   \EndIf
   \EndFor
   \EndFor
   \EndFor
   \State {\bfseries Output:} $\ub_l(u)$ for all $l=\{t,t+1,\dots,n-1\}$, $u\in V$
\end{algorithmic}
\end{algorithm}

\noindent {\bf Tunable nonbacktracking lower bounds (tNB-LB)}: We first order the vertex set as $V'=\{v_1,\dots,v_n\}$, which satisfies $d(S_0,v_i)\leq d(S_0,v_j)$, for every $i < j$.
Given a non-negative integer parameter $t\leq n$, we obtain a $t$-size subnetwork $IC(G[V_t],{\cal P}[V_t],S_0\cap V_t)$, where 
$G[V_t]$ is the vertex-induced subgraph which is induced by the set of nodes $V_t=\{v_1,\dots, v_{t}\}$, and ${\cal P}[V_t]$ is the corresponding transmission probability matrix. For each $v_i\in V_t$, we compute the exact probability $p_t(v_i)$ that node $v_i$ is infected in the subnetwork $IC(G[V_t],{\cal P}[V_t],S_0\cap V_t)$. 
Then, we start (non-parameterized) \nblb\ algorithm from $k=t+1$ with the new initial condition: for all $k\leq t$,
\begin{eqnarray}
\lb(v_k)=p_t(v_k).
\end{eqnarray}
Finally, tNB-LB computes the lower bound as $\sum_{v_k\in V'}\lb(v_k)$.

\begin{algorithm}[H]
   \caption{Tunable NB-LB (tNB-LB)}
   \label{alg:lowerknob}
\begin{algorithmic}
   \State {\bfseries parameter:} non-negative integer $t\leq n$
   \State {\bfseries Initialize:} $\lowinf = 0$
   \For{$k=1$ {\bfseries to} $t$}
   \State $\lb(v_k)=p_t(v_k)$
   \State $\lowinf \pluseq \lb(v_k)$
   \For{$v_i\in \{N^+(v_k)\cap\{v_j:j>t\}\}$}
   \State $\lowin(v_i).\minsert((\lb(v_k),{\cal P}'_{v_kv_i}))$
   \EndFor
   \EndFor
   \For{$k=t+1$ {\bfseries to} $n$}
   \State $\lb(v_k)=\lowprob(\lowin(v_k))$
   \State $\lowinf \pluseq \lb(v_k)$
   \For{$v_i\in N^+(v_k)\setminus S_0$}
   \State $\lowin(v_i).\minsert((\lb(v_k),{\cal P}'_{v_kv_i}))$
   \EndFor
   \EndFor
   \State {\bfseries Output:} $\lowinf$
\end{algorithmic}
\end{algorithm}

For a larger $t$, the algorithm results in tighter lower bounds. However, the computational complexity may increase exponentially with respect to $t$, the size of the subnetwork. This algorithm can adopt Monte Carlo simulations on the subnetwork to avoid the large computational complexity. However, this modification results in probabilistic lower bounds, rather than theoretically guaranteed lower bounds. Nonetheless, this can still give a significant improvement, because the Monte Carlo simulations on a smaller size of network require less computations to stabilize the estimation.

\section{Experimental results}

In this section, we evaluate the \nbub\ and \nblb\ in independent cascade models on a variety of classical synthetic networks.

$\bf{Network\ generation:}$
We consider $4$ classical random graph models with the parameters shown as follows: Erdos Renyi random graphs with $ER(n=1000,p=3/1000)$, scale-free networks $SF(n=1000,\alpha=2.5)$, random regular graphs $Reg(n=1000,d=3)$, and random tree graphs with power-law degree distributions $T(n=1000,\alpha=3)$.
For each graph model, we generate $100$ networks $IC(G,pA,\{s\})$ as follows. 
The graph $G$ is the largest connected component of a graph drawn from the graph model, the seed node $s$ is selected randomly from the vertex set of $G$, and $A$ is the adjacency matrix of $G$. Then, the corresponding IC model has the same transmission probability $p$ for every edge.

$\bf{Evaluation\ of\ bounds:}$
For each network generated, we compute the following quantities for each $p\in \{0.1,0.2,\dots,0.9\}$.
\begin{itemize}[leftmargin=0.5cm,topsep=-4pt,parsep=0pt]
    \item[$\cdot$] $\sigma_{mc}$: the estimation of the influence with $1000000$ Monte Carlo simulations.
    \item[$\cdot$] $\upinf$: the upper bound obtained by \nbub.
    \item[$\cdot$] $\upinf_{spec}$: the spectral upper bound by~\cite{lemonnier2014tight}.
    \item[$\cdot$] $\lowinf$: the lower bound obtained by \nblb.
    \item[$\cdot$] $\lowinf_{prob}$: the probabilistic lower bound obtained by $10$ Monte Carlo simulations.
\end{itemize}

\begin{figure*}[h]
\centering
\includegraphics[width=\textwidth]{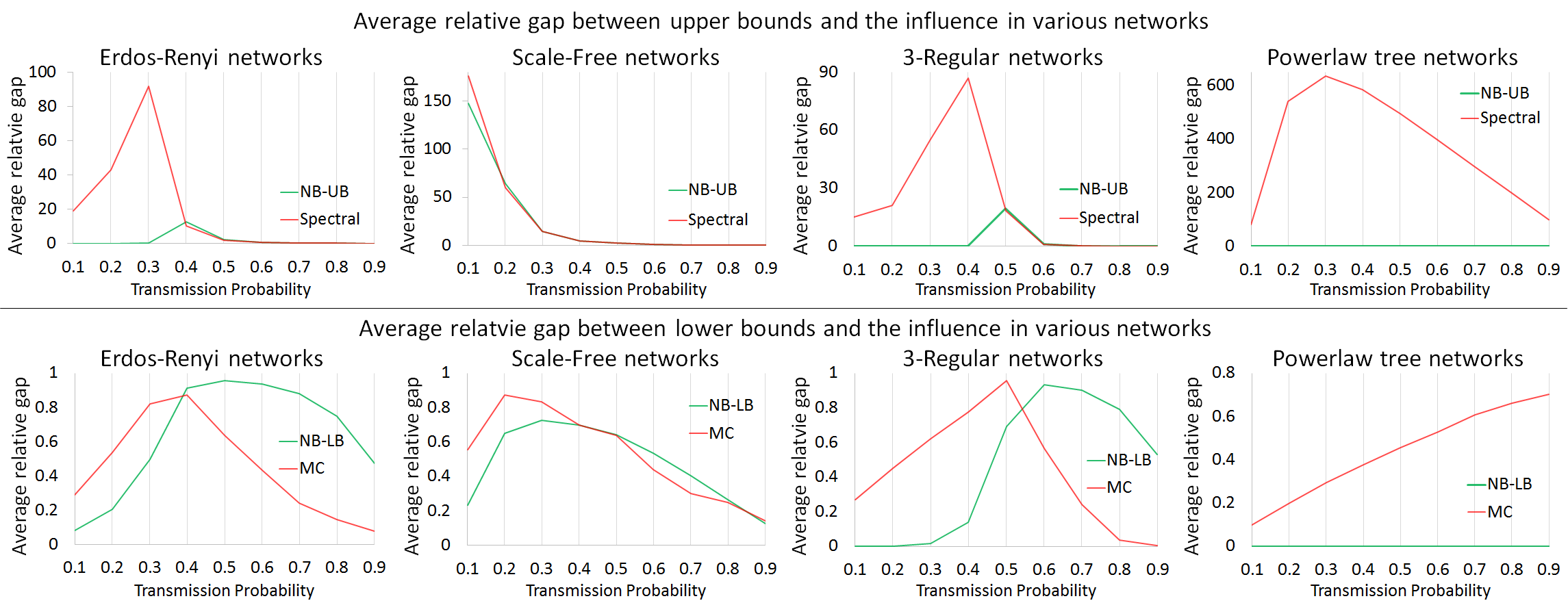}
\caption{This figure compares the average relative gap of the bounds: \nbub, the spectral upper bound in~\cite{lemonnier2014tight}, \nblb, and the probabilistic lower bound computed by MC simulations, for various types of networks.}
\label{fig:networks}
\end{figure*}

The probabilistic lower bound is chosen for the experiments, since there has not been any tight lower bound. The sample size of $10$ is determined to overly match the computational complexity of \nblb\ algorithm.  
In Figure~\ref{fig:networks}, we compare the average relative gap of the bounds for every network model and for each transmission probability, where the true value is assumed to be $\sigma_{mc}$. For example, the average relative gap of \nbub\ for $100$ Erdos Renyi networks $\{{\cal N}_i\}_{i=1}^{100}$ with the transmission probability $p$ is computed by $\frac{1}{100}\sum_{i\in [100]}\frac{\upinf[{\cal N}_i]-\sigma_{mc}[{\cal N}_i]}{\sigma_{mc}[{\cal N}_i]}$, 
where $\upinf[{\cal N}_i]$ and $\sigma_{mc}[{\cal N}_i]$ denote the \nbub\ and the MC estimation, respectively, for the network ${\cal N}_i$.

$\bf{Results:}$
Figure~\ref{fig:networks} shows that \nbub\ outperforms the upper bound in~\cite{lemonnier2014tight} for the Erdos-Renyi and random $3$-regular networks, and performs comparably for the scale-free networks. Also, \nblb\ gives tighter bounds than the MC bounds on the Erdos-Renyi, scale-free, and random regular networks when the transmission probability is small, $p<0.4$. The \nbub\ and \nblb\ compute the exact influence for the tree networks since both algorithms avoid backtracking walks. 

Next, we show the bounds on exemplary networks.

\subsection{Upper bounds}\label{sec:exp_up}

\begin{figure}
    \centering
    \begin{subfigure}[b]{0.45\textwidth}
        \includegraphics[width=\textwidth]{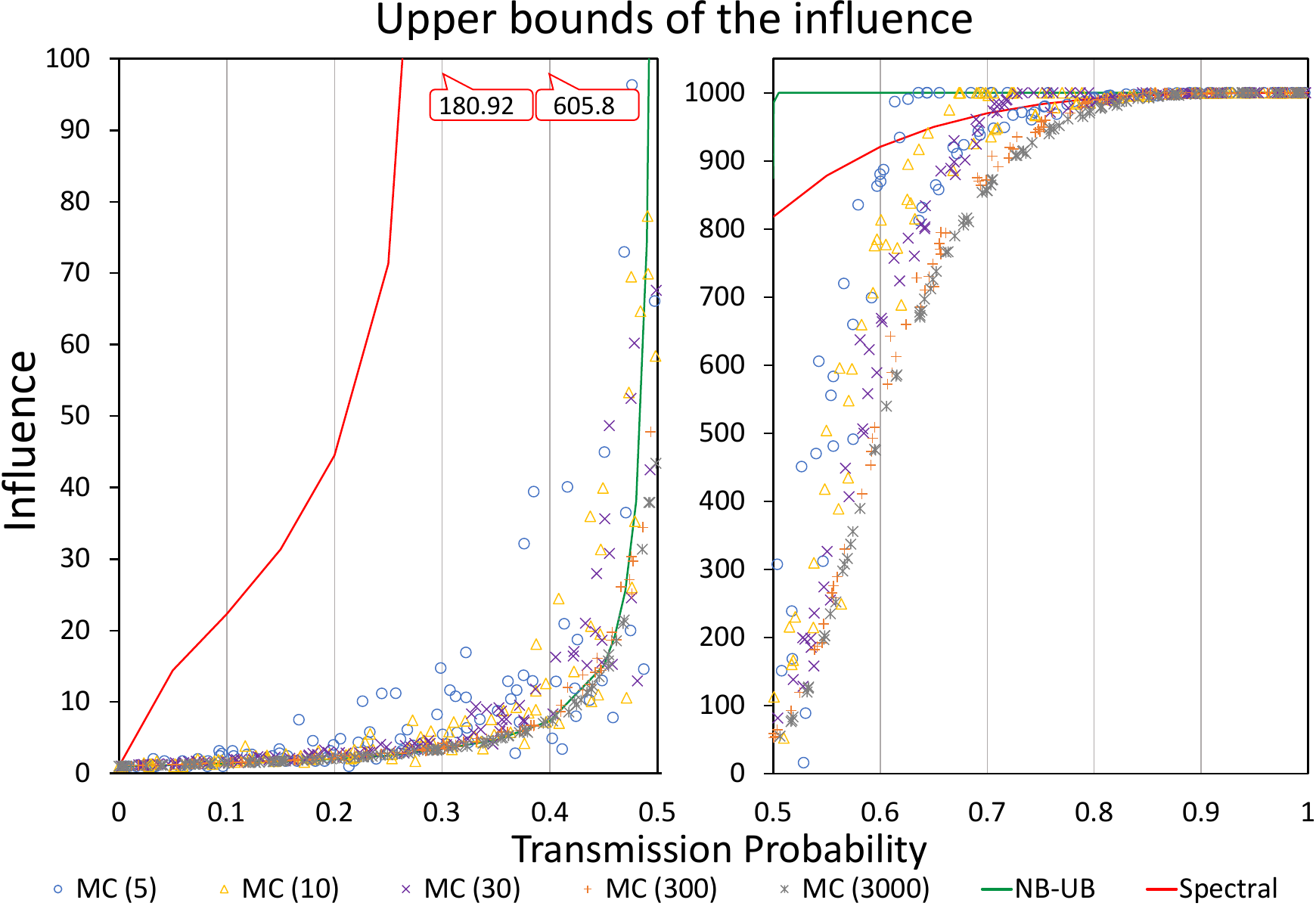}
        \caption{}
        \label{fig:up}
    \end{subfigure}
    ~ 
    \begin{subfigure}[b]{0.45\textwidth}
        \includegraphics[width=\textwidth]{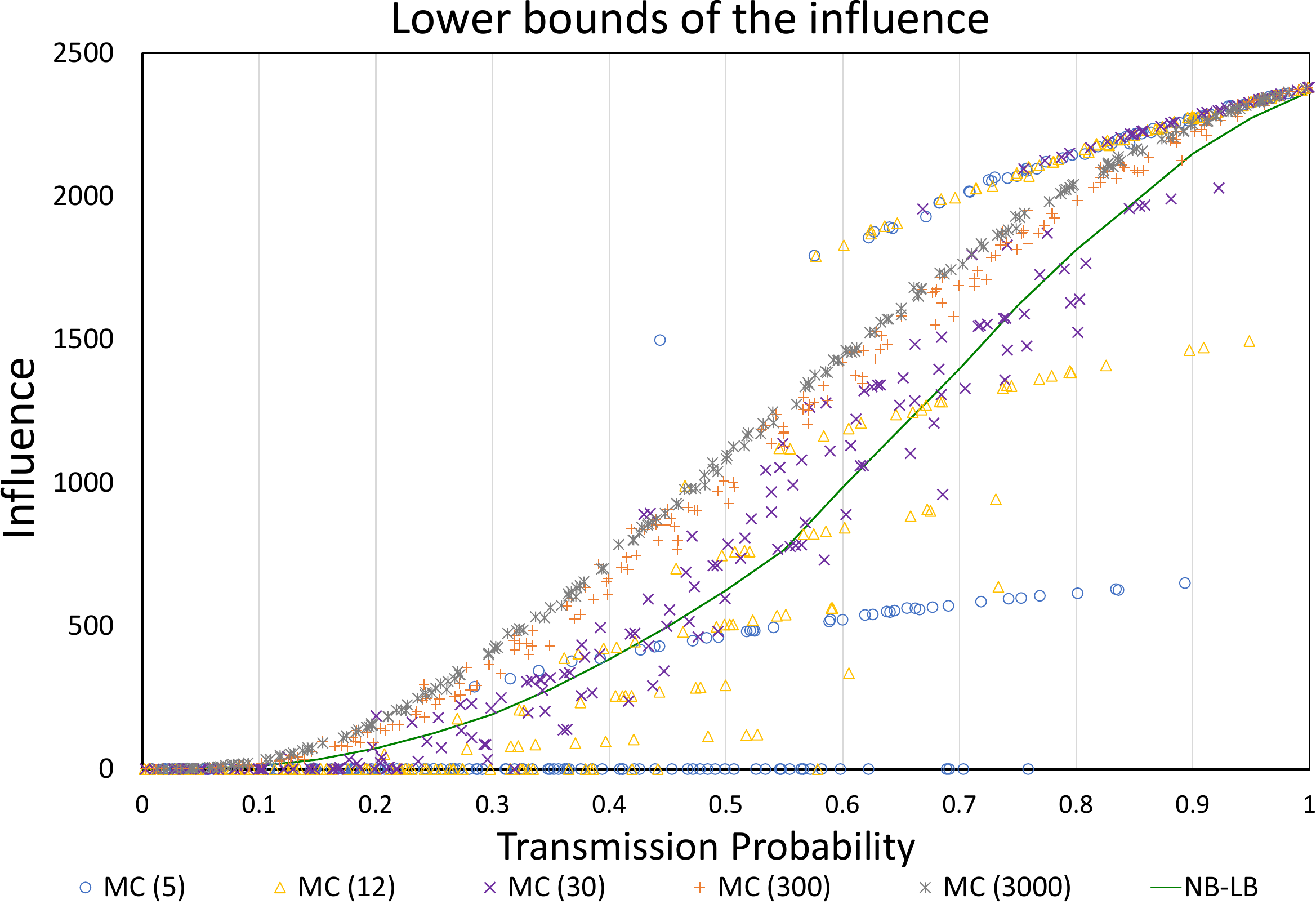}
        \caption{}
        \label{fig:low}
    \end{subfigure}
    \caption{(a) The figure compares various upper bounds on the influence in the $3$-regular network in section~\ref{sec:exp_up}. The MC upper bounds are computed with various simulation sizes and shown with the data points indicated with $\mathrm{MC}(N)$, where $N$ is the number of simulations. The spectral upper bound in~\cite{lemonnier2014tight} is shown in red line, and \nbub\ is shown in green line. 
    \\(b) The figure shows lower bounds on the influence of a scale-free network in section~\ref{sec:exp_low}. The probabilistic lower bounds shown with points are obtained from Monte Carlo simulations with various simulation sizes, and the data points indicated with $\mathrm{MC}(N)$ are obtained by $N$ number of simulations. \nblb\ is shown in green line.}\label{fig:animals}
\end{figure}

$\bf{Selection\ of\ networks:}$
In order to illustrate a typical behavior of the bounds, we have chosen the network in Figure~\ref{fig:up} as follows. 
First, we generate $100$ random $3$-regular graphs $G$ with $1000$ nodes and assign a random seed $s$.
Then, the corresponding IC model is defined as $IC(G,{\cal P}=pA,S_0=\{s\})$, where $A$ is the adjacency matrix, resulting in the same transmission probability $p$ for every edge.
For each network, we compute \nbub\ and MC estimation of $1000$ simulations.
We note that with $1000$ simulations, the estimation given by MC cannot guarantee stability. Yet, given limited time and resources, it is acceptable for the selection process.
Then, we compute the score for each network. The score is defined as the sum of the square differences between the upper bounds and MC estimations over the transmission probability $p\in \{0.1,0.2,\dots,0.9\}$. Finally, a graph whose score is the median from all $100$ scores is chosen for Figure~\ref{fig:up}.

$\bf{Results:}$
In figure~\ref{fig:up}, we compare 1) the upper bounds introduced~\cite{lemonnier2014tight} and 2) the probabilistic upper bounds obtained by Monte Carlo simulations with $99\%$ confidence level, to \nbub. The upper bounds from MC simulations are computed with the various sample sizes $N\in\{5, 10, 30, 300, 3000\}$. It is evident from the figure that a larger sample size provides a tighter probabilistic upper bound.
\nbub\ outperforms the bound by~\cite{lemonnier2014tight} and the probabilistic MC bound when the transmission probability is relatively small. Further, it shows a similar trend as the MC simulations with a large sample size.

\subsection{Lower bounds}\label{sec:exp_low}
$\bf{Selection\ of\ networks:}$
We adopt a similar selection process as in the selection for upper bounds, but with the scale free networks, with $3000$ nodes and $\alpha=2.5$. 

$\bf{Results:}$
We compare probabilistic lower bounds obtained by MC with $99\%$ confidence level to \nblb. The lower bounds from Monte Carlo simulations are computed with various sample sizes $N\in\{5, 12, 30, 300, 3000\}$, which accounts for a constant, $\log(n)$, $0.01n$, $0.1n$, and $n=|V|$.
\nblb\ outperforms the probabilistic bounds by MC with small sample sizes. 
Recall that the computational complexity of the lower bound in algorithm~\ref{alg:lower} is $O(|V|+|E|)$, which is the computational complexity of a constant number of Monte Carlo simulations.
In figure~\ref{fig:low}, it shows that \nblb\ is tighter than the probabilistic lower bounds with the same computational complexity, and it also agrees with the behavior of the MC simulations.

\subsection{Tunable bounds}
In this section, we present tunable upper and lower bounds on example networks. 

\begin{figure}[H]
    \centering
    \begin{subfigure}[b]{0.45\textwidth}
        \includegraphics[width=\textwidth]{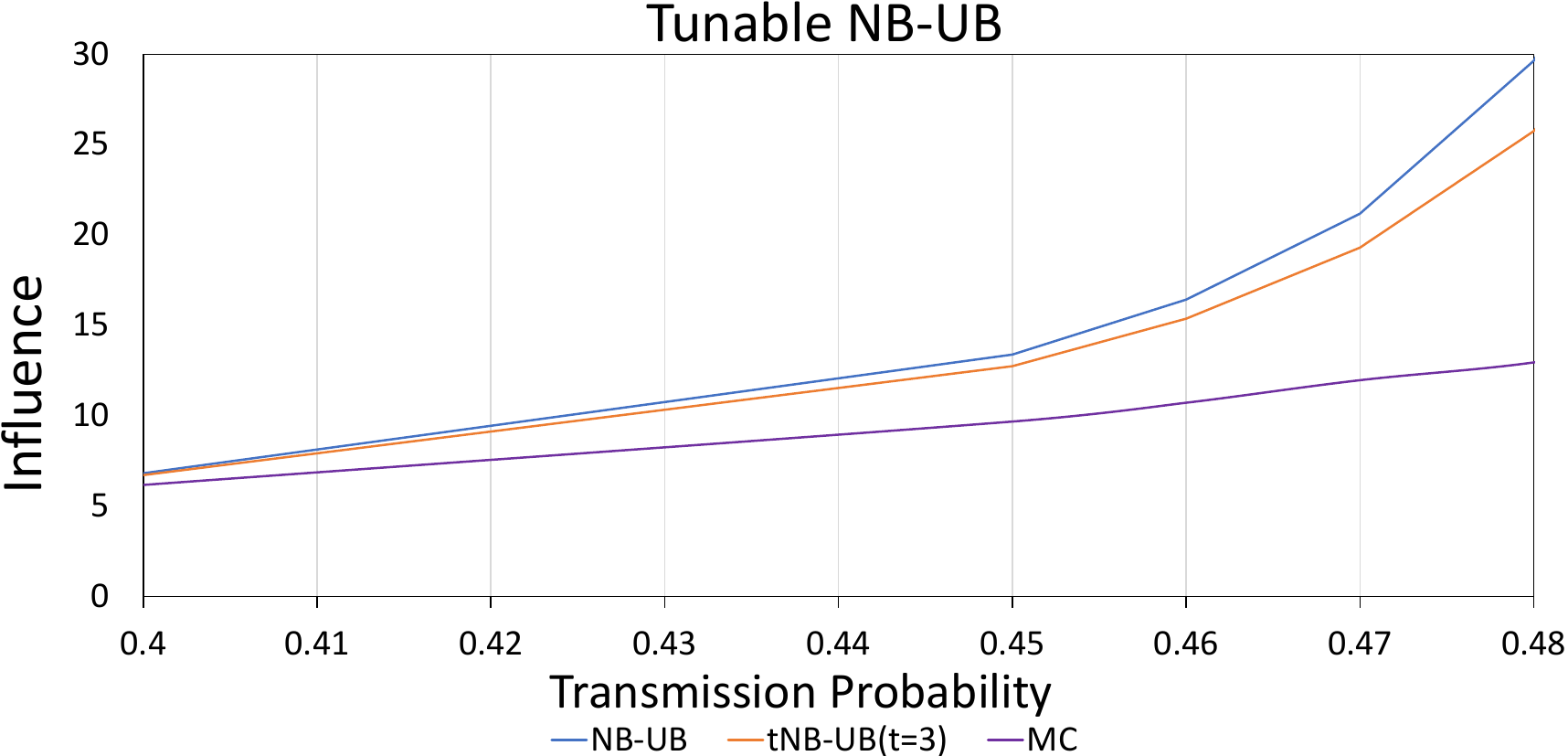}
        \caption{}
        \label{fig:tunaup}
    \end{subfigure}
    ~ 
    \begin{subfigure}[b]{0.45\textwidth}
        \includegraphics[width=\textwidth]{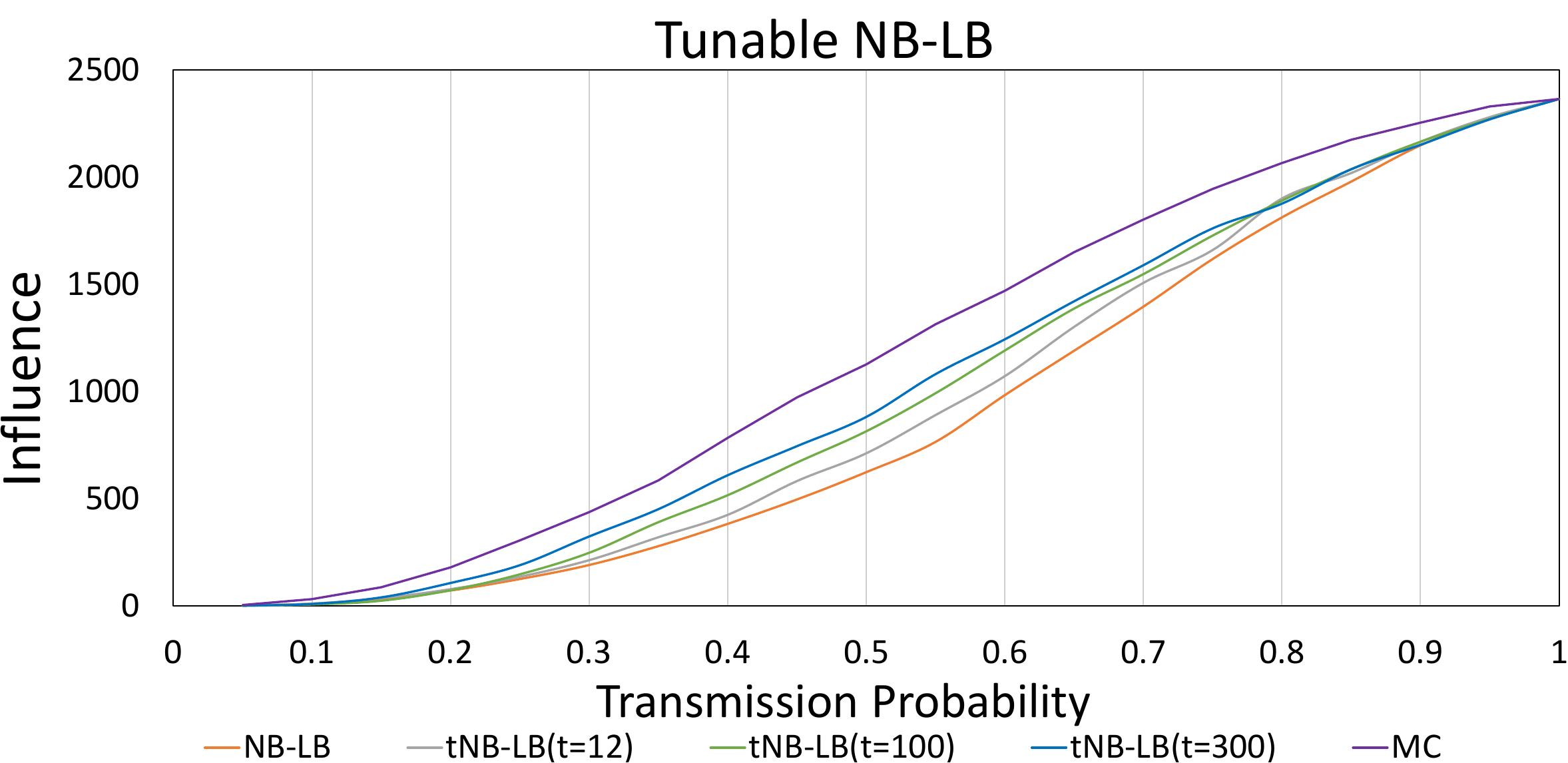}
        \caption{}
        \label{fig:tunalow}
    \end{subfigure}
    \caption{(a) \nbub, t\nbub\ with $t=3$, and MC estimation with $10000$ simulations on a $3$-regular network with $100$ nodes.
    \\(b) \nblb, t\nblb\ with various $t\in\{12,100,300\}$, and MC estimation with $3000000$ simulations on a scale-free network with $3000$ nodes.}\label{fig:tunafig}
\end{figure} 

In Figure~\ref{fig:tunaup}, we show t\nbub\ on a sample network. We consider a $3$-regular network with $100$ nodes and a single seed. Since the \nbub\ gives a tight bound on $p<0.4$, we plot t\nbub\ on $p\in (0.4,0.5)$ where it shows some improvements with small $t$.

In Figure~\ref{fig:tunalow}, we present t\nblb\ on a scale-free network with $3000$ nodes, $\alpha=2.5$, and a single seed. We compare t\nblb\ with various choices of $t\in\{1,12,100,300\}$, and t\nblb\ approaches the MC estimation as $t$ grows.

\section{Conclusion}
In this paper, we propose both upper and lower bounds on the influence in the independent cascade models, and provide algorithms to efficiently compute the bounds. We extend the results by proposing tunable bounds which can adjust the trade off between the efficiency and the accuracy. Finally, the tightness and the performance of bounds are shown with experimental results. One can further improve the bounds considering $r$-nonbacktracking walks, i.e. avoiding cycles of length $r$ rather than just backtracks, and we leave this for future study.

\section{Acknowledgements}
We would like to thank Colin Sandon for helpful discussions. This research was partly supported by the NSF CAREER Award CCF-1552131 and the ARO grant W911NF-16-1-0051.

\pagebreak
\bibliographystyle{abbrv}
\bibliography{main}

\end{document}